\documentclass[12pt]{article}

\usepackage{apacite,latexsym,graphics,epsfig,amsmath,amssymb}
\usepackage{authblk}
\usepackage{natbib}
\usepackage{enumerate}
\usepackage{amsthm}
\usepackage{url}

\usepackage[OT2,OT1]{fontenc}
\newcommand\cyr{%
\renewcommand\rmdefault{wncyr}%
\renewcommand\sfdefault{wncyss}%
\renewcommand\encodingdefault{OT2}%
\normalfont
\selectfont}
\DeclareTextFontCommand{\textcyr}{\cyr}

\theoremstyle{definition}
\newtheorem{example}{Example}[section]
\newtheorem{definition}{Definition}[section]

\theoremstyle{plain}
\newtheorem{theorem}{Theorem}[section]

\newtheorem{corollary}[theorem]{Corollary}

\begin{document}

\title{Relational models for contingency tables}

\author[*]{Anna Klimova}
\author[**]{Tam\'{a}s Rudas}
\affil[*]{University of Washington, Seattle, WA}
\affil[**]{E\"{o}tv\"{o}s Lor\'{a}nd University, Budapest, Hungary}
\author[*]{Adrian Dobra}

\maketitle

\begin{abstract}
The paper considers general multiplicative models for complete and incomplete contingency tables that generalize log-linear and several other models and are entirely coordinate free. Sufficient conditions of the existence of maximum likelihood estimates under these models are given, and it is shown that the usual equivalence between multinomial and Poisson likelihoods holds if and only if an overall effect is present in the model. If such an effect is not assumed, the model becomes a curved exponential family and a related mixed parameterization is given that relies on non-homogeneous odds ratios. Several examples are presented to illustrate the properties and use of such models.
\end{abstract}

\section*{Introduction}

The main objective of the paper is to develop a new class of models for the set of all strictly positive distributions on contingency tables and on some sets of cells that have a more general structure. The proposed relational models are motivated by traditional log-linear models, quasi models and some other  multiplicative models for discrete distributions that have been discussed in the literature. 

Under log-linear models \citep{BFH}, cell probabilities are determined by multiplicative effects associated with various subsets of the variables in the contingency table. However, some cells may have other characteristics in common, and there always has been interest in models that also allow for multiplicative effects that are associated with those characteristics. Examples, among others, include quasi models \citep{Goodman68,Goodman72}, topological models \citep{Hauser1978, Hout}, indicator models \citep{Zelt1992}, rater agreement-disagreement models \citep{TannerRaters, TannerDisagr}, two-way subtable sum models \citep{HTY09}. All these models, applied in different contexts, have one common idea behind them. A model is generated by a class of subsets of cells, some of which may not be induced by marginals of the table, and, under the model, every cell probability is the product of effects associated with subsets the cell belongs to. This idea is generalized in the relational models framework. 

The outline of the paper is as follows. The definition of a table and the definition of a relational model generated by a class of subsets of cells in the table are given in Section \ref{sec1}. The cells are characterized by strictly positive parameters (probabilities or intensities); a table is a structured set of cells. Under the model, the parameter of each cell is the product of effects associated with the subsets in the generating class, to which the cell belongs. Two examples are given to illustrate this definition. Example \ref{CIex} shows how traditional log-linear models fit into the framework and Example \ref{CrabEx} describes how multiplicative models for incomplete contingency tables are handled. 

The degrees of freedom and the dual representation of relational models are discussed in Section \ref{secDual}. Every relational model can be stated in terms of generalized odds ratios. The minimal number of generalized odds ratios required to specify the model is equal to the number of degrees of freedom in this model. 

The models for probabilities that include the overall effect and all relational models for intensities are regular exponential families. Under known conditions \citep[cf.][]{Barndorff1978}, the maximum likelihood estimates for cell frequencies exist and are unique; the observed values of canonical statistics are equal to their expected values. If the overall effect is not present, a relational model for probabilities forms a curved exponential family. The maximum likelihood estimates in the curved case exist and are unique under the same condition as for regular families; the observed values of canonical statistics are proportional to their expected values. The maximum likelihood estimates for cell frequencies under a model for intensities and under a model for probabilities, when the model matrix is the same, are equal if and only if the model for probabilities is a regular family. These facts are proved in Section \ref{secPoisson}.

A mixed parameterization of finite discrete exponential families is discussed in Section \ref{sec3}. Any relational model is naturally defined under this parameterization: the corresponding generalized odds ratios are fixed and the model is parameterized by remaining mean-value parameters. The distributions of observed values of subset sums and generalized odds ratios are variation independent and, in the regular case, specify the table uniquely. 

Two applications of the framework are presented in Section \ref{SecApplic}. These are the analysis of social mobility data and the analysis of a valued network with given attributes. These two examples suggest that the flexibility of the framework and substantive interpretation of parameters make relational models appealing for many settings.

\section{Definition and Log-linear Representation of Relational Models}\label{sec1}

Let $Y_1, \dots, Y_K$ be the discrete random variables modeling certain characteristics of the population of interest. Denote the domains of the variables by $\mathcal{Y}_1, \dots, \mathcal{Y}_K$ respectively. 
A point $(y_1, y_2,\dots,y_K) \in \mathcal{Y}_1 \times \dots \times \mathcal{Y}_K$ generates a cell if and only if the outcome $(y_1, y_2,\dots,y_K)$ appears in the population. A cell $(y_1, y_2,\dots,y_K)$ is called empty if the combination is not included in the design. 

Let $\mathcal{I}$ denote the lexicographically ordered set of non-empty cells in $\mathcal{Y}_1 \times \dots \times \mathcal{Y}_K$, and $|\mathcal{I}|$ denote the cardinality of $\mathcal{I}$. Since the case, when $\mathcal{I} = \mathcal{Y}_1 \times \dots \times \mathcal{Y}_K$, corresponds to a classical complete contingency table, then the set $\mathcal{I}$ is also called a table. 

Depending on the procedure that generates data on $\mathcal{I}$, the population may be characterized by cell probabilities or cell intensities. The parameters of the true distribution will be denoted by $\boldsymbol \delta =\{\delta(i),\,\,\mbox{for } i \in \mathcal{I}\}$. In the case of probabilities, $\delta(i) = p(i) \in (0,1)$, where $\sum_{i \in \mathcal{I}} p(i) = 1$; in the case of intensities, $\delta(i) = \lambda(i) > 0$. Let $\mathcal{P}$ denote the set of strictly positive  distributions parameterized by $\boldsymbol \delta$.

\begin{definition} \label{DEFrm}
Let $\mathbf{S} = \{S_1, \dots, S_J\}$, be a class of non-empty subsets of the table $\mathcal{I}$, $\mathbf{A}$ a $J \times |\mathcal{I}|$ matrix with entries
\begin{equation} \label{genMMentries}
a_{ji}=\mathbf{I}_{j}(i)=\left\{\begin{array}{l} 1, \mbox{ if  the } i\mbox{-th cell is in } S_j,  \\ 0, \mbox{ otherwise,}  \end{array}\right.
\mbox{ for } i = 1, \dots,|\mathcal{I}| \mbox{ and } \, j = 1,\dots,J .
\end{equation}
\textit{A  relational model} $RM(\mathbf{S})\subseteq \mathcal{P}$ with the model matrix $\mathbf{A}$ is the subset of $\mathcal{P}$ satisfying the equation:
\begin{equation} \label{PMmatr}
\mbox{log } \boldsymbol \delta = \mathbf{A}'\boldsymbol \beta,
\end{equation}
for some $\boldsymbol \beta \in \mathbb{R}^J$.
\end{definition}

Under the model (\ref{PMmatr}) the parameters of the distribution can also be written as 
\begin{equation} \label{PMmatr1}
\delta(i) =\mbox{exp  }\{\sum_{j=1}^J \mathbf{I}_j(i) \beta_j\} = \prod_{j=1}^J (\theta_j)^{\mathbf{I}_j(i)},
\end{equation}
where $\theta_j = \mbox{exp  }(\beta_j)$, for $j =1, \dots, J$. 

The parameters $\boldsymbol \beta$ in (\ref{PMmatr}) are called the log-linear parameters. The parameters $\boldsymbol \theta$ in (\ref{PMmatr1}) are called the multiplicative parameters. If the subsets in $\mathbf{S}$ are cylinder sets, the parameters $\boldsymbol \beta$ coincide with the parameters of the corresponding log-linear model.

In the case $\boldsymbol \delta = \boldsymbol p$ it must be assumed that $\cup_{j=1}^JS_j = \mathcal{I}$, i.e. there are no zero columns in the matrix $\mathbf{A}$. A zero column implies that one of the probabilities is 1 under the model and the model is thus trivial. 

The example below describes a model of conditional independence as a relational model. 
\begin{example} \label{CIex}
Consider the model of conditional independence $[Y_1Y_3][Y_2Y_3]$ of three binary variables $Y_1$, $Y_2$, $Y_3$, each taking values in $\{0,1\}$. The model is expressed as 
$$p_{ijk}=\frac{p_{i+k}p_{+jk}}{p_{++k}},$$
where $p_{i+k}, p_{+jk}, p_{++k}$ are marginal probabilities in the standard notation \citep{BFH}.
Let $\mathbf{S}$ be the class consisting of the cylindrical sets associated with the empty marginal and the marginals $Y_1$, $Y_2$ $Y_3$, $Y_1Y_3$, $Y_2Y_3$. The model matrix computed from (\ref{genMMentries}) is not full row rank and thus the model  parameters are not identifiable (cf. Section \ref{secDual}). A full row rank model matrix can be obtain by setting, for instance, the level $0$ of each variable as the reference level. After that, the model matrix is equal to      
\begin{equation} \label{mmRef} 
\mathbf{A}=\left(
\begin{array} {cccccccc}
1& 1& 1& 1& 1& 1& 1& 1\\
0& 0& 0& 0& 1& 1& 1& 1\\
0& 0& 1& 1& 0& 0& 1& 1\\
0& 1& 0& 1& 0& 1& 0& 1\\
0 &0& 0& 0& 0& 1& 0& 1\\
0& 0& 0& 1& 0& 0& 0& 1
\end{array}\right)
\end{equation} 
The first row corresponds to the cylindrical set associated with the empty marginal. The next three rows correspond to the cylindrical sets generated by the level $1$ of $Y_1$, $Y_2$, $Y_3$ respectively. The fifth row corresponds to the cylindrical set generated by the level $1$ for both $Y_1$ and $Y_3$, and the last row - to the cylindrical set corresponding to the level $1$ for both $Y_2$ and $Y_3$. \qed
%$\Box$
\end{example}

In the next example, one of the cells in the Cartesian product of the domains of the variables is empty and the sample space $\mathcal{I}$ is a proper subset of this product. 

\begin{example} \label{CrabEx} 
The study described by \cite{Kawamura1995} compared three bait types for trapping swimming crabs: fish alone, sugarcane alone, and sugarcane-fish combination. During the experiment, catching traps without bait was not considered. Three Poisson random variables are used to model the amount of crabs caught in the three traps.
The notation for the intensities is shown in Table \ref{FishIntens}.
The model assuming that there is a multiplicative effect of using both bait types at the same time will be tested in this paper. The hypothesis of interest is 
\begin{equation}\label{CrabsModelEq}
\lambda_{00}= \lambda_{01} \lambda_{10}.
\end{equation}

The effect can be tested using the relational model for rates on the class $\mathbf{S}$ consisting of two subsets: $\mathbf{S}=\{S_1, S_2\}$, where $S_1 = \{(0,0), (0,1)\}$ and $S_2 = \{(0,0), (1,0)\}$: 
\begin{equation*}
\mbox{log }\boldsymbol \lambda = {\mathbf{A}'}\boldsymbol\beta,
\end{equation*}
Here, the model matrix $$\mathbf{A}=
\left(
\begin{array}{ccc}
1& 1 & 0\\
1& 0 & 1\\
\end{array}
\right),$$ and $\boldsymbol \beta = (\beta_1,\beta_2)'$. The relationship between the two forms of the model will be explored in the next section. \qed
%$\Box$
\end{example}

\begin{table}
\centering
\caption{Poisson intensities by bait type.}
\label{FishIntens}
\begin{center}
\begin{tabular}{|c|cc|}
\hline
 &
\multicolumn{2}{|c|}{ \mbox{ Fish}}\\
\cline{2-3}
\mbox{ Sugarcane }& \mbox{Yes}& \mbox{No }\\
\hline
\mbox{ Yes }	& $\lambda_{00}$	& $\lambda_{01}$\\
\mbox{ No } 	& $\lambda_{10}$  &   -  \\ 
\hline
\end{tabular}
\end{center}
\end{table}

\begin{table}[h]
\begin{minipage}[b]{0.5\linewidth}
\centering
\caption{Number of trapped \textit{Charybdis japonica} by bait type.}
\label{Fish2}
\vspace{5mm}
\begin{tabular}{|c|cc|}
\hline
 &
\multicolumn{2}{|c|}{ \mbox{ Fish}}\\
\cline{2-3}
\mbox{ Sugarcane }& \mbox{Yes}& \mbox{No }\\
\hline
\mbox{ Yes }	& 36 	& 2 \\
\mbox{ No } 	& 11  &   -  \\ 
\hline
\end{tabular}
\end{minipage}
\hspace{0.5cm}
\begin{minipage}[b]{0.5\linewidth}
\centering
\caption{Number of trapped \textit{Portunuspelagicus} by bait type.}
\label{Fish1}
\vspace{5mm}
\begin{tabular}{|c|cc|}
\hline
 &
\multicolumn{2}{|c|}{ \mbox{ Fish}}\\
\cline{2-3}
\mbox{ Sugarcane }& \mbox{Yes}& \mbox{No }\\
\hline
\mbox{ Yes }	& 71 	& 3 \\
\mbox{ No } 	& 44  &   -  \\ 
\hline
\end{tabular}
\end{minipage}
\end{table}

\section{Parameterizations and Degrees of Freedom}\label{secDual}

A choice of subsets in $\mathbf{S} = \{S_1, \dots, S_J\}$ is implied by the statistical problem, and the relational model $RM(\mathbf{S})$ can be parameterized with different model matrices, which may be useful depending on substantive meaning of the model. Sometimes a particular choice of subsets leads to a model matrix $\mathbf{A}$ with linearly dependent rows and thus non-identifiable model parameters. To ensure identifiability, a reparameterization, that is sometimes referred to as model matrix coding, is needed. Examples of frequently used codings are reference coding, effects coding, orthogonal coding, polynomial coding \citep[cf.][]{Christensen}. 

Write $R(\mathbf{A})$ for the row space of $\mathbf{A}$ and call it the design space of the model. The elements of $R(\mathbf{A})$ are $|\mathcal{I}|$-dimensional row-vectors and $\boldsymbol 1$ denotes the row-vector with all components equal to $1$. Reparameterizations of the model have form $\boldsymbol\beta = \mathbf{C}\boldsymbol \beta_1$, where $\boldsymbol \beta_1$ are the new parameters of the model and $\mathbf{C}$ is a $J \times [rank (\mathbf{A})]$ matrix such that the modified model matrix $\mathbf{C}'\mathbf{A}$ has a full row rank and $R(\mathbf{A})$ = $R(\mathbf{C}'\mathbf{A})$. 
Then $R(\mathbf{A})^{\bot}$ = $R(\mathbf{C}'\mathbf{A})^{\bot}$, that is $Ker \,(\mathbf{A})$ = $Ker \,(\mathbf{C}'\mathbf{A})$.

Let $\mathcal{P} = \mathcal{P}_{\boldsymbol \delta} = \{P_{\boldsymbol \delta}: \,\, \boldsymbol \delta \in \mathcal{N}\}$ be the set of all positive distributions on the table $\mathcal{I}$. Here the parameter space $\mathcal{N}$ is an open subset of $\mathbb{R}^{|\mathcal{I}|}$. Suppose $\Theta \subset \mathcal{N}$. Then the set $\mathcal{P}_{0} = \{P_{\boldsymbol \delta}: \,\, \boldsymbol \delta \in \Theta \subset \mathcal{N}\}$ is a model in $\mathbf{P}_{\boldsymbol \delta}$. The number of degrees of freedom of the model $\mathcal{P}_{0}$ is the difference between dimensionalities of $\mathcal{N}$ and $\Theta$.

\begin{theorem}
The number of degrees of freedom in a relational model $RM(\mathbf{S})$ is $|\mathcal{I}| - dim R(\mathbf{A})$. 
\end{theorem} 

\begin{proof}

Let $\boldsymbol \delta = \boldsymbol p=(p(1), \dots,p(|\mathcal{I}|)'$. Since $\sum_{i \in \mathcal{I}} p(i) = 1$, then the parameter space $\mathcal{N}$ is $|\mathcal{I}|-1$-dimensional. If $RM(\mathbf{S})$ is a relational model for probabilities (\ref{PMmatr1}), its multiplicative parameters $\boldsymbol \theta$ must satisfy the normalizing equation
\begin{equation} \label{NormEqTheta}
\sum_{i \in \mathcal{I}}  \prod_{j=1}^J (\theta_j)^{\mathbf{I}_j(i)} = 1.
\end{equation}
Since the model matrix is full row rank, then the set $\Theta = \{\boldsymbol \theta \in \mathbb{R}^J_+: \,\, \sum_{i \in \mathcal{I}}  \prod_{j=1}^J (\theta_j)^{\mathbf{I}_j(i)} = 1  \}$ is a $J-1$-dimensional surface in $\mathbb{R}^J$. 
Therefore, the number of degrees of freedom of $RM(\mathbf{S})$ is $dim \mathcal{N}-dim \Theta = |\mathcal{I}| -1-(J-1) = |\mathcal{I}| - dim R(\mathbf{A})$.

Let $\boldsymbol \delta = \boldsymbol \lambda$ and $RM(\mathbf{S})$ is a model for intensities. In this case, $\mathcal{N}=\{\boldsymbol \lambda \in \mathbb{R}^{|\mathcal{I}|}_+\} $ and $\Theta \subset \mathcal{N}$ consists of all $\boldsymbol \lambda$ satisfying (\ref{PMmatr1}). Since no normalization is needed,  $dim \mathcal{N} = |\mathcal{I}|$ and $dim \Theta = dim R(\mathbf{A})$ and thence the number of degrees of freedom of $RM(\mathbf{S})$ is equal to $|\mathcal{I}| - dim R(\mathbf{A})$.
\end{proof}

The theorem implies that the number of degrees of freedom of the relational model coincides with $dim\,Ker(\mathbf{A})$. This is in coherence with the fact that the kernel of the model matrix is invariant of reparameterizations of the model (\ref{PMmatr}). To restrict further analysis to models with a positive number of degrees of freedom suppose in the sequel that $Ker(\mathbf{A})$ is non-trivial. Without loss of generality, suppose further that the model matrix is full row rank. 

\begin{definition}
A matrix $\mathbf{D}$ with rows that form a basis of $Ker(\mathbf{A})$ is called \textit{a kernel basis matrix} of the relational model $RM(\mathbf{S})$. 
\end{definition}

The representation (\ref{PMmatr}) is a primal (intuitive) representation of relational models; a dual representation is described in the following theorem.

\begin{theorem} \label{mixedTh}
{\begin{enumerate}[(i)]
\item The distribution, parameterized by $\boldsymbol \delta$, belongs to the relational model $RM(\mathbf{S})$ if and only if 
\begin{equation} \label{rmD}
\mathbf{D} \mbox{log } \boldsymbol \delta = \boldsymbol 0.
\end{equation}
\item The matrix $\mathbf{D}$ may be chosen to have integer entries.
\end{enumerate}}
\end{theorem}

\begin{proof} 
\begin{enumerate}[(i)]
\item By the definition of a relational model, 
$$P_{\boldsymbol \delta} \in RM(\mathbf{S}) \,\, \Leftrightarrow \,\, \mbox{log } \boldsymbol \delta  = \mathbf{A}'\boldsymbol \beta.$$ 
The orthogonality of the design space and the null space implies that $\mathbf{A}\mathbf{D}' = \boldsymbol 0$ for any kernel basis matrix $\mathbf{D}$. The rows of $\mathbf{D}$ are linearly independent. Therefore  
$$P_{\boldsymbol \delta} \in RM(\mathbf{S}) \,\, \Leftrightarrow \,\, \mathbf{D} \mbox{log } \boldsymbol \delta = \mathbf{D} \mathbf{A}'\boldsymbol \beta = \boldsymbol 0.$$ 

\item 
Since $\mathbf{A}$ has full row rank, then the dimension of $Ker\, (\mathbf{A})$ is equal to $K_0 = |\mathcal{I}|-J$. 

By Corollary 4.3b \citep[pg. 49]{Schrijver}, there exists a unimodular matrix $\mathbf{U}$, i.e. $\mathbf{U}$ is integer and $det\,\mathbf{U} = \pm 1$, such that $\mathbf{A}\mathbf{U}$ is the Hermite normal form of $\mathbf{A}$, that is 
{\begin{enumerate}[(a)]
\item $\mathbf{A}\mathbf{U}$ has form  $[\mathbf{B}, \mathbf{0}]$,
\item $\mathbf{B}$ is a non-negative, non-singular, lower triangular matrix;
\item $\mathbf{A}\mathbf{U}$ is an $n\times m$ matrix with entries $c_{ij}$ such that  $c_{ij} < c_{ii}$ for all $i=1, \dots,n, \,\,j=1, \dots, m,\,\, i\ne j$.
\end{enumerate}
}

Let $\mathbf{I}_{K_0}$ stand for the $K_0\times K_0 $ identity matrix, $\mathbf{0}$ denote the $J \times K_0$ zero matrix, and  $\mathbf{Z}$ be the following $|\mathcal{I}| \times K_0$ matrix: 
$$\mathbf{Z} = \left(\begin{array}{c}\mathbf{0}\\ \mathbf{I}_{K_0}\end{array}\right).$$

Since the matrix $\mathbf{A}\mathbf{U}$ has form $[\mathbf{B}, \mathbf{0}]$ where $\mathbf{B}$ is the nonsingular, lower triangular, $J \times J$ matrix, then 
$(\mathbf{A}\mathbf{U})\mathbf{Z} = \boldsymbol 0$. 

Set $\mathbf{D}' =\mathbf{U}\mathbf{Z}$. Then
\begin{equation}\label{Ceq}
\mathbf{A}\mathbf{D}' = \mathbf{A}\mathbf{U}\mathbf{Z} = \boldsymbol 0. 
\end{equation}
The matrix $\mathbf{U}$ is integer and nonsingular, the columns of $\mathbf{Z}$ are linearly independent. Therefore the matrix $\mathbf{D}'$ is integer and has linearly independent columns. Hence the matrix $\mathbf{D}$ is an integer kernel basis matrix of the model.
\qedhere
\end{enumerate}
\end{proof}

\textbf{ Example  \ref{CIex} (Revisited)}
{
For the model of conditional independence $dim\,  Ker(\mathbf{A}) = 2$. If the kernel basis matrix is chosen as $$\mathbf{D}=
\left(
\begin{array}{cccccccc}
1&0 &-1&0& -1&0& 1 & 0 \\
0&1&0 &-1&0& -1&0& 1  
\end{array}
\right),
$$
the equation $\mathbf{D} \mbox{log } \boldsymbol p =  \boldsymbol 0$ is equivalent to the following constraints: 
$$\frac{p_{000}p_{110}}{p_{010}p_{100}} =1, \,\,\, \frac{p_{001}p_{111}}{p_{011}p_{101}} = 1.$$
The latter is a well-known representation of the model $[Y_1Y_3][Y_2Y_3]$ in terms of the conditional odds ratios \citep{BFH}.  \qed %$\Box$
}

The dual representation (\ref{rmD}) of a relational model is, in fact, a model representation in terms of some monomials in $\boldsymbol \delta$. All types of polynomial expressions that may arise in the dual representation of a relational model are captured by the following definition.  
\begin{definition} 
Let $u(i), v(i) \in \mathbb{Z}_{\geq 0}$ for all $i \in \mathcal{I}$, $\boldsymbol \delta ^{\boldsymbol u} = \prod_{i \in \mathcal{I}}\delta(i)^{u(i)}$ and $\boldsymbol \delta ^{\boldsymbol v} = \prod_{i \in \mathcal{I}}\delta(i)^{v(i)}$.
\textit{A generalized odds ratio} for a positive distribution, parameterized by $\boldsymbol \delta$, is a ratio of two monomials: 
\begin{equation} \label{defOR}
\mathcal{OR} = \boldsymbol \delta ^{\boldsymbol u}/\boldsymbol \delta ^{\boldsymbol v}.
\end{equation}
\end{definition} 
The odds ratio $\mathcal{OR} = \frac{\boldsymbol{\delta}^{\boldsymbol u}}{\boldsymbol{\delta}^{\boldsymbol v}}$ is called homogeneous if  $\sum_{i \in \mathcal{I}} u(i) =  \sum_{i \in \mathcal{I}}v(i)$.

To express a relational model $RM(\mathbf{S})$ in terms of generalized odds ratios, write the rows $\boldsymbol{d}_1, \boldsymbol{d}_2, \dots, \boldsymbol{d}_{K_0} \in  \mathbb{Z}^{|\mathcal{I}|}$ of a kernel basis matrix $\mathbf{D}$ in terms of their positive and negative parts:
\begin{equation*}
\boldsymbol{d}_l = \boldsymbol{d}_l^{+} - \boldsymbol{d}_l^{-}, 
\end{equation*}
where $\boldsymbol{d}_l^{+},\, \boldsymbol{d}_l^{-} \geq \boldsymbol 0$ for all $l = 1,2, \dots, K_0$.  
Then the model (\ref{rmD}) takes form
\begin{equation*} 
\boldsymbol{d}_l^{+} \mbox{log }\boldsymbol \delta = \boldsymbol{d}_l^{-}\mbox{log }\boldsymbol \delta, \,\, \mbox{for } l = 1,2, \dots, K_0,
\end{equation*}
which is equivalent to the model representation in terms of generalized odds ratios:
\begin{equation}\label{ORparamz}
\boldsymbol \delta^{\boldsymbol{d}_l^{+}} / \boldsymbol \delta^{\boldsymbol{d}_l^{-}} = 1, \,\, \mbox{for } l = 1,2, \dots, K_0.
\end{equation}
The number of degrees of freedom is equal to the minimal number of generalized odds ratios required to uniquely specify a relational model.

\vspace{7 mm}
\textbf{ Example  \ref{CrabEx}  (Revisited)}
{The model $\lambda_{00} =  \lambda_{01} \lambda_{10}$ can be expressed in the matrix form as:
\begin{equation}
\mathbf{D} \mbox{log } \boldsymbol \lambda = 0,
\end{equation}
where $\mathbf{D}= (1, -1 , -1)$. The matrix $\mathbf{D}$ is a kernel basis matrix of the relational model,
as one would expect. Finally, the model representation in terms of generalized odds ratios is
$$\frac{\lambda_{00}}{\lambda_{01} \lambda_{10}} = 1.$$
\qed %$\Box$
}
\vspace{5 mm}

The role of generalized odds ratios in parameterizing distributions in $\mathcal{P}$ will be explored in Section \ref{sec3}.

\section{Relational Models as Exponential Families: Poisson vs Multinomial Sampling}\label{secPoisson}

The representation (\ref{PMmatr1}) implies that a relational model is an exponential family of distributions. The canonical parameters of a relational model are $\beta_j$'s and the canonical statistics are indicators of subsets $\mathbf{I}_j$. 
Relational models for intensities and relational models for probabilities are considered in this section in more detail.

Let $RM_{\boldsymbol\lambda}(\mathbf{S})$ denote a relational model for intensities and $RM_{\boldsymbol p}(\mathbf{S})$ denote a relational model for probabilities with the same model matrix $\mathbf{A}$, that has a full rank $J$.   

If the distribution of a random vector $\boldsymbol Y$ is parameterized by intensities $\boldsymbol \lambda$, then, under the model $RM_{\boldsymbol\lambda}(\mathbf{S})$,  
\begin{equation}\label{efP}
P(\boldsymbol{Y}=\boldsymbol{y}) = \frac{1}{\prod_{i \in \mathcal{I}}y(i)!}\mbox{exp }\{ \boldsymbol \beta'\mathbf{A} \boldsymbol y - \boldsymbol 1 \mbox{exp}\mathbf{A}'\boldsymbol \beta \}. 
\end{equation}

If the distribution of $\boldsymbol Y$ is multinomial, with parameters N and $\boldsymbol p$, then, under the model $RM_{\boldsymbol p}(\mathbf{S})$,  
\begin{equation}\label{efMultinomial}
P(\boldsymbol{Y}=\boldsymbol{y}) = \frac{N!}{\prod_{i \in \mathcal{I}}y(i)!}\mbox{exp }\{ \boldsymbol \beta'\mathbf{A} \boldsymbol y \}. 
\end{equation}

Set 
\begin{equation} \label{Suff}
\boldsymbol T(\boldsymbol Y) = \mathbf{A} \boldsymbol Y = (\mathcal{T}_1(\boldsymbol Y), \mathcal{T}_2(\boldsymbol Y),\dots,\mathcal{T}_J(\boldsymbol Y))'. 
\end{equation}
For each $j \in 1, \dots, J$, the statistic $\mathcal{T}_j(\boldsymbol Y) = \sum_{{i}\in\mathcal{I}} \mathbf{I}_j({i})Y({i})$ is the subset sum corresponding to the subset $S_j$.

\begin{theorem}\label{PoissonAsEF}
A model $RM_{\boldsymbol\lambda}(\mathbf{S})$ is a regular exponential family of order $J$.
\end{theorem}

\begin{proof}
 The model matrix $\mathbf{A}$ has full rank; no normalization is needed for intensities. Therefore, the representation (\ref{efP}) is minimal and the exponential family is regular, of order $J$.
\end{proof}

Relational models for probabilities may have a more complex structure than relational models for 
intensities and, in some cases, become curved exponential families \citep{Efron1975, BrownBook, KassVos}.

\begin{theorem}\label{MultAsEF}
If $\boldsymbol{1} \in R(\mathbf{A})$, a model $RM_{\boldsymbol p}(\mathbf{S})$ is a regular exponential family of order $J- 1$; otherwise, it is a curved exponential family of order $J - 1$. 
\end{theorem}

\begin{proof}

Suppose that $\boldsymbol{1} \in R(\mathbf{A})$. Without loss of generality, $\mathcal{I} = S_1 \in \mathbf{S}$ and thus 
\begin{equation}\label{pfEF1}
P(\boldsymbol Y  = \boldsymbol y) = \frac{N!}{\prod_{i \in \mathcal{I}}y(i)!}  \mbox{exp } \{N \beta_1 +
\sum_{j=2}^J (\sum _{i \in \mathcal{I}}y(i)\mathbf{I}_j(i)) \beta_j\}.
\end{equation}
The exponential family representation given by (\ref{pfEF1}) is minimal; the model $RM_{\boldsymbol p}(\mathbf{S})$ is a regular exponential family of order $J- 1$.

If $\boldsymbol{1} \notin R(\mathbf{A})$ then, independent of parameterization, the model matrix does not include the row of all $1$s. The normalization is required and thus the parameter space is a manifold of the dimension $J-1$ in $\mathbb{R}^J$ \citep[see e.g.][p.229]{Rudin}. In this case, $RM_{\boldsymbol p}(\mathbf{S})$ is a curved exponential family of the order $J-1$ \citep{KassVos}.
\end{proof}

If a relational model is a regular exponential family, the maximum likelihood estimate of the canonical parameter exists if and only if the observed value of the canonical statistic is contained in the interior of the convex hull of the support of its  distribution \citep{BarndorffCox1994}. In this case, the MLE is also unique.

It is well known for log-linear models that, when the total sample size is fixed, the kernel of the likelihood is the same for the multinomial and Poisson sampling scheme and thus the maximum likelihood estimates of the cell frequencies, obtained under either sampling scheme, are equal \citep[see e.g.][p.448]{BFH}. The following theorem is an extension of this result.

\begin{theorem} \label{ThPoissonMequiv}
Assume that, for a given set of observations, the maximum likelihood estimates $\hat{\boldsymbol \lambda}$, under the model $RM_{\boldsymbol \lambda}(\mathbf{S})$, and $\hat{\boldsymbol p}$, under the model $RM_{\boldsymbol p}(\mathbf{S})$, exist. 
The following four conditions are equivalent:
{\begin{enumerate}[(A)]
\item \label{A} The MLEs for cell frequencies obtained under either model are the same.
\item \label{B} Vector $\boldsymbol 1$ is in the design space $R(\mathbf{A})$.
\item \label{C} Both models may be defined by homogeneous odds ratios.
\item \label{D} The model for intensities is scale invariant.
\end{enumerate}}
\end{theorem}

\begin{proof}

(\ref{A}) $\Longleftarrow$ (\ref{B})

The maximum likelihood estimates for probabilities, under the model $RM_{\boldsymbol p}(\mathbf{S})$, satisfy the likelihood equations 
\begin{eqnarray} \label{MultLkE}
\mathbf{A}\boldsymbol y &=& {\alpha} \mathbf{A}\hat{\boldsymbol p}\\
\boldsymbol 1\hat{\boldsymbol p} &=& 1. \nonumber
\end{eqnarray}
Here $\alpha$ is the Lagrange multiplier. 

If $\boldsymbol 1 \in R(\mathbf{A})$ then there exists a $\boldsymbol k \in \mathbb{R}^J$ such that $\boldsymbol k' \mathbf{A}=  \boldsymbol 1$. Multiplying both sides of the first equation in (\ref{MultLkE}) by $\boldsymbol k'$ yields ${\alpha} = N$ and hence
\begin{equation}\label{LikeEqalpha41}
\mathbf{A} \boldsymbol y = N \mathbf{A}\hat{\boldsymbol p}.
\end{equation}

The maximum likelihood estimates for intensities, under $RM_{\boldsymbol \lambda}(\mathbf{S})$, satisfy the likelihood equations 
\begin{eqnarray}\label{LikeP41}
\mathbf{A}\boldsymbol y = \mathbf{A}\hat{\boldsymbol \lambda}.
\end{eqnarray}

From the equations (\ref{LikeEqalpha41}) and (\ref{LikeP41}):
$$\hat{\boldsymbol \lambda} - N\hat{\boldsymbol p} \in Ker \, \mathbf{A}.$$
The latter implies that $\boldsymbol 1 (\hat{\boldsymbol \lambda} - N\hat{\boldsymbol p}) = 0$ and $N = \boldsymbol 1 \hat{\boldsymbol \lambda}$. Therefore  
$$\hat{\boldsymbol p} = \frac{\hat{\boldsymbol \lambda}}{\boldsymbol 1 \hat{\boldsymbol \lambda}}$$
and the maximum likelihood estimates for cell frequencies obtained under either model are the same:
$$\hat{\boldsymbol y} = N\hat{\boldsymbol p} = \hat{\boldsymbol \lambda}.$$

(\ref{A}) $\Longrightarrow$ (\ref{B})

Suppose that $\hat{\boldsymbol y} = N\hat{\boldsymbol p} = \hat{\boldsymbol \lambda}$. Under the model $RM_{\boldsymbol \lambda}(\mathbf{S})$
$$\mbox{log } (\hat{\boldsymbol \lambda}) = \mathbf{A}'\hat{\boldsymbol \beta}_1$$
for some $\hat{\boldsymbol \beta}_1$.
On the other hand, under the model $RM_{\boldsymbol p}(\mathbf{S})$,
$$\mbox{log } (\hat{\boldsymbol \lambda}) = \mbox{log } (N\hat{\boldsymbol p}) = \mathbf{A}'\hat{\boldsymbol \beta}_2 + \mbox{log }N\boldsymbol 1'$$
for some $\hat{\boldsymbol \beta}_2$. 
The condition $\mathbf{A}'\hat{\boldsymbol\beta}_1 = \mathbf{A}'\hat{\boldsymbol\beta}_2 + \mbox{log }N\boldsymbol 1'$
can only hold if $\boldsymbol 1 \in R(\mathbf{A})$.

(\ref{B}) $\Longleftrightarrow$ (\ref{C})

The vector $\boldsymbol 1 \in R(\mathbf{A})$ if and only if all rows of a kernel basis matrix $\mathbf{D}$ are orthogonal to  $\boldsymbol 1$ and the sum of entries in every row of $\mathbf{D}$ is zero. The latter is equivalent to the generalized odds ratios obtained from rows of $\mathbf{D}$ being homogeneous. 

(\ref{D}) $\Longleftrightarrow$ (\ref{B})

Let $t > 0, \, t \ne 1$. 
\begin{eqnarray*}
\mathbf{D} \mbox{log } (t\boldsymbol{\lambda}) = \boldsymbol 0 \Longleftrightarrow 
\mbox{log } t \cdot (\mathbf{D} \boldsymbol 1') = \boldsymbol 0 \Longleftrightarrow  
\mathbf{D} \boldsymbol 1' = \boldsymbol 0, \mbox{ or } \boldsymbol 1 \in R(\mathbf{A}). 
\end{eqnarray*}

\end{proof}

\begin{corollary}\label{CorEq}
For a given set of observations, the MLEs of the subset sums under a model $RM_{\boldsymbol p}(\mathbf{S})$ are equal to their observed values if and only if $\boldsymbol 1 \in R(\mathbf{A})$.
\end{corollary}

\begin{proof}
If $\boldsymbol 1 \in R(\mathbf{A})$ the model $RM_{\boldsymbol p}(\mathbf{S})$ is a regular exponential family. The subset sums are canonical statistics; their MLEs are the same as observed.
 
Suppose that the MLEs of the subset sums are equal to their observed values. Then $N\mathbf{A} \boldsymbol{p} = N\mathbf{A} \boldsymbol {\hat{p}}$ and thus $\boldsymbol{p} - \hat{\boldsymbol{p}} \in Ker \mathbf{A}$.
Since $\sum(p(i) - \hat{p}(i)) = \sum p(i) - \sum\hat{p}(i) = 1 - 1 = 0$, vector $\boldsymbol{p} - \hat{\boldsymbol{p}}$ is orthogonal to $\boldsymbol 1$ and thus $\boldsymbol 1 \in R(\mathbf{A})$. 
\end{proof}

\begin{corollary}\label{CorProp}
Suppose $\boldsymbol 1 \notin R(\mathbf{A})$. For a given set of observations, the MLEs, if exist, of the subset sums under a model $RM_{\boldsymbol p}(\mathbf{S})$ are proportional to their observed values.
\end{corollary}

\begin{proof}
In this case the value of $\alpha$ cannot be found from (\ref{MultLkE}) and one can only assert that
$$\mathbf{A}\boldsymbol y = \frac{\alpha}{N} \mathbf{A}\hat{\boldsymbol y}.$$
\end{proof}

Example \ref{CrabEx} illustrates a situation when a relational model for intensities is not scale invariant. This model is a curved exponential family. The existence and uniqueness of the maximum likelihood estimates in such relational models is proved next.

\begin{theorem}
Let  $\boldsymbol Y \sim \mbox{M }(N, \boldsymbol p)$, $\boldsymbol y $ be a realization of $\boldsymbol Y$, and $RM_{\boldsymbol p}(\mathbf{S})$ be a relational model, given $\boldsymbol 1 \notin R(\mathbf{A})$. The maximum likelihood estimate for $\boldsymbol p$, under the model $RM_{\boldsymbol p}(\mathbf{S})$, exists and unique if and only if $\boldsymbol T(\boldsymbol y) >0$.
\end{theorem}

\begin{proof}

A point in the canonical parameter space of the model $RM_{\boldsymbol p}(\mathbf{S})$ that maximizes the log-likelihood subject to the normalization constraint is a solution to the optimization problem:
$$ \underset{\mbox{s.t. } \boldsymbol \beta \in \mathcal{D}}{\mbox{max } l(\boldsymbol \beta; \boldsymbol y),}$$
where $$l(\boldsymbol \beta; \boldsymbol y)= \mathcal{T}_1(\boldsymbol y)\beta_1 + \dots + \mathcal{T}_J(\boldsymbol y)\beta_J$$ and 
$$\mathcal{D}=\{\boldsymbol \beta \in \mathbb{R}^J_{-}:\,\, \sum_{i \in \mathcal{I}} \mbox{exp}\{\sum_{j=1}^J \mathbf{I}_j(i)\beta_j\}-1=0\}.$$
The set $\mathcal{D}$ is non-empty and is a level set of a convex function. The level sets of convex functions are not convex in general. However the sub-level sets of convex functions and hence the set 
$$\mathcal{D}_{\leq}=\{\boldsymbol \beta \in \mathbb{R}^J_{-}:\,\, \sum_{i \in \mathcal{I}} \mbox{exp}\{\sum_{j=1}^J \mathbf{I}_j(i)\beta_j\}-1\leq 0\}$$
are convex.

The set of maxima of $l(\boldsymbol \beta; \boldsymbol y)$ over the set $\mathcal{D}_{\leq}$ is nonempty and consists of a single point if and only if \citep[][Section~3]{Bertsekas2009}
$$R_{\mathcal{D}_{\leq}} \cap R_{-l} = L_{\mathcal{D}_{\leq}}\cap L_{-l}.$$
Here $R_{\mathcal{D}_{\leq}}$ is the recession cone of the set $\mathcal{D}_{\leq}$, $R_{-l}$ is the recession cone of the function  $-l$, $L_{\mathcal{D}_{\leq}}$ is the lineality space of $\mathcal{D}_{\leq}$, and $L_{-l}$ is the lineality space of $-l$. 

The recession cone of $\mathcal{D}_{\leq}$ is the orthant $\mathbb{R}^J_{-}$, including the origin; the lineality space is  $L_{\mathcal{D}_{\leq}}=\{0\}$. The lineality space of the function $-l$ is the plane passing through the origin, with the normal $\mathbf{T}(\boldsymbol y)$; the recession cone of $-l$ is the half-space above this plane. The condition $R_{\mathcal{D}_{\leq}} \cap R_{-l} = L_{\mathcal{D}_{\leq}}\cap L_{-l} = \{0\}$ holds if and only if all components of $\mathbf{T}(\boldsymbol y)=(\mathcal{T}_1(\boldsymbol y), \dots,  \mathcal{T}_J(\boldsymbol y))'$ are positive.

The function $l(\boldsymbol \beta; \boldsymbol y)$ is linear; its maximum is achieved on $\mathcal{D}$. Therefore there exists one and only one $\boldsymbol \beta$ which maximizes the likelihood over the canonical parameter space and the maximum likelihood estimate for $\boldsymbol p$, under the model $RM_{\boldsymbol p}(\mathbf{S})$, exists and unique. 
\end{proof}

\begin{table}[h]
\begin{minipage}[b]{0.5\linewidth}
\centering
\caption{The MLEs for the Number of trapped \textit{Charybdis japonica} by bait type}
\label{Fish2MLE}
\vspace{5mm}
\begin{tabular}{|c|cc|}
\hline
 &
\multicolumn{2}{|c|}{ \mbox{ Fish}}\\
\cline{2-3}
\mbox{ Sugarcane }& \mbox{Yes}& \mbox{No }\\
\hline
\mbox{ Yes }	& 35.06 	& 2.94 \\
\mbox{ No } 	& 11.94  &   -  \\ 
\hline
\end{tabular}
\end{minipage}
\hspace{0.5cm}
\begin{minipage}[b]{0.5\linewidth}
\centering
\caption{The MLEs for the Number of trapped \textit{Portunuspelagicus} by bait type.}
\label{Fish1MLE}
\vspace{5mm}
\begin{tabular}{|c|cc|}
\hline
 &
\multicolumn{2}{|c|}{ \mbox{ Fish}}\\
\cline{2-3}
\mbox{ Sugarcane }& \mbox{Yes}& \mbox{No }\\
\hline
\mbox{ Yes }	& 72.31 	& 1.69 \\
\mbox{ No } 	& 42.69  &   -  \\ 
\hline
\end{tabular}
\end{minipage}
\end{table}

\vspace{7 mm}
\textbf{ Example  \ref{CrabEx}  (Revisited)}
{
In this example, the relational model for intensities is not scale invariant. The maximum likelihood estimates for the cell frequencies exist and are shown in Tables \ref{Fish2MLE} and \ref{Fish1MLE}. The observed Pearson's statistics are $X^2 = 0.40$ and $X^2 = 1.07$ respectively, on one degree of freedom. \qed %$\Box$
}
\vspace{5 mm}

The relational models framework deals with models generated by subsets of cells, and the model matrix for a relational model is an indicator matrix that has only 0-1 entries. Theorems \ref{mixedTh}, \ref{ThPoissonMequiv} hold if the model matrix has non-negative integer entries.  The next example illustrates how the techniques and theorems apply to some discrete exponential models.  

\begin{example} \label{ExCalves}
This example, given in \citep{Agresti2002}, describes the study carried out to determine if a pneumonia infection has an immunizing effect on dairy calves. Within 60 days after birth, the calves were exposed to a pneumonia infection. The calves that got the infection were then classified according to whether or not they got the secondary infection within two weeks after the first infection cleared up. The number of the infected calves is thus a random variable with the multinomial distribution $M(N,(p_{11},p_{12},p_{22})')$, where $N$ denotes the total number of calves in the sample. Suppose further that $p_{11}$ is the probability to get both the primary and the secondary infection, $p_{12}$ is the probability to get only the primary infection and not the secondary one, and $p_{22}$ is the probability not to catch either the primary or the secondary infection.
Let $0 < \pi < 1$ denote the probability to get the primary infection. The hypothesis of no immunizing effect of the primary infection is expressed as \citep[cf.][]{Agresti2002}
\begin{equation} \label{calP}
p_{11}=\pi^2, \,\,p_{12}=\pi(1-\pi), \,\,p_{22}=1 - \pi.
\end{equation} 
Since the model (\ref{calP}) is also expressed in terms of non-homogeneous odds ratios: $$\frac{p_{11}  p_{22}^2}{p_{12}^2} = 1,$$ 
then it is a relational model for probabilities, without the overall effect.

Write $N_{11}, \, N_{12},\, N_{22}$ for the number of calves in each category and $n_{11}, \, n_{12}, \, n_{22}$ for their realizations. The log-likelihood is proportional to
$$(2n_{11}+n_{12})\mbox{log } \pi +  (n_{12}+n_{22}) \mbox{log } (1 - \pi).$$
The canonical statistic $\boldsymbol T = (\mathcal{T}_1,\mathcal{T}_2) = (2N_{11}+N_{12}, N_{12}+N_{22})$ is two-dimensional; the canonical parameter space $\{(\mbox{log } \pi, \mbox{log }(1- \pi)): \,\, \pi \in (0,1) \}$ is the curve in $\mathbb{R}^2$  shown on Figure \ref{FigureCalves}. The model (\ref{calP}) is thus a curved exponential family of order 1. 

The likelihood is maximized by   
$$\hat{\pi} = \frac{2n_{11}+n_{12}}{2n_{11}+2n_{12}+n_{22}}= \frac{{T}_1}{{T}_1+{T}_2},$$
where $T_1 = 2n_{11}+n_{12}$ and $T_2 = n_{12}+n_{22}$ are observed components of the canonical statistic, or subset sums. The MLEs of the subset sums can be expressed in terms of their observed values as 
\begin{eqnarray*}
\hat{\mathcal{T}_1} &=& N(2 \hat{\pi}^2 + \hat{\pi}(1-\hat{\pi})) = N(\frac{2{T}_1^2}{({T}_1+{T}_2)^2} + \frac{{T}_1{T}_2}{({T}_1+{T}_2)^2}) = {T}_1\frac{N(2{T}_1+{T}_2)}{({T}_1+{T}_2)^2},\\
\hat{\mathcal{T}_2} &=& N(\hat{\pi}(1-\hat{\pi})+(1-\hat{\pi})) = N(\frac{{T}_1{T}_2}{({T}_1+{T}_2)^2} + \frac{{T}_2}{{T}_1+{T}_2}) = {T}_2\frac{N(2{T}_1+{T}_2)}{({T}_1+{T}_2)^2}.\\
\end{eqnarray*}
Thus, under the model (\ref{calP}), the MLEs of the subset sums differ from their observed values by the factor $\frac{N(2\mathcal{T}_1+\mathcal{T}_2)}{(\mathcal{T}_1+\mathcal{T}_2)^2}$. For the data and the MLEs in Table \ref{CalvesData}, this factor is approximately $0.936$. \qed %$\Box$
\end{example}

\begin{figure}
\centering
\includegraphics[scale=0.55]{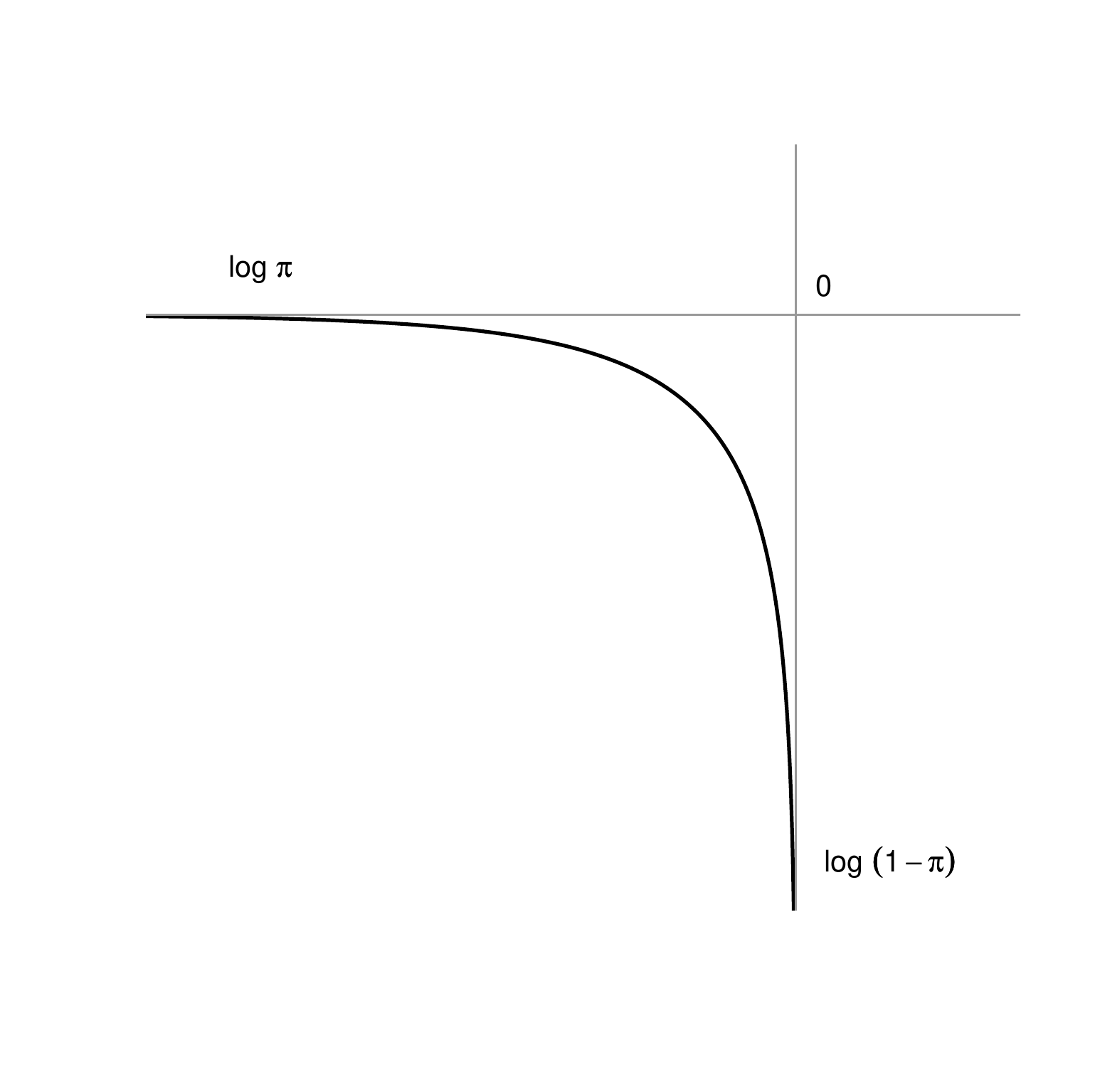}
\caption{The canonical parameter space in Example \ref{ExCalves}.}
\label{FigureCalves}
\end{figure}

\begin{table}[h]
\centering
\caption{Observed Counts for Primary and Secondary Pneumonia Infection of Calves. 
The MLEs are shown in parentheses \citep{Agresti2002}. }
\label{CalvesData}
\vspace{5mm}
\begin{tabular}{|c|cc|}
\hline
 &
\multicolumn{2}{|c|}{ \mbox{ Secondary Infection}}\\
\cline{2-3}
\mbox{ Primary Infection }& \mbox{Yes}& \mbox{No }\\
\hline
\mbox{ Yes }	& 30 (38.1)	& 63 (39.0)\\
\mbox{ No } 	& -  &  63 (78.9 )  \\ 
\hline
\end{tabular}
\end{table}

\section{Mixed Parameterization of Exponential Families} \label{sec3}

Let $\mathcal{P}_{\boldsymbol \delta}$ be an exponential family formed by all strictly positive distributions on $\mathcal{I}$ and $\mbox{log }\boldsymbol \delta$ be the canonical parameters of this family. Denote by $\mathcal{P}_{\boldsymbol \gamma}$  the reparameterization of $\mathcal{P}_{\boldsymbol \delta}$ defined by the following one-to one mapping:
\begin{equation}\label{param}  
\mbox{log }\boldsymbol \delta = {\mathbf{M}'}\boldsymbol\gamma,
\end{equation}
where $\mathbf{M}$ is a full rank, $|\mathcal{I}| \times |\mathcal{I}|$, integer matrix, and $\boldsymbol\gamma \in \mathbb{R}^{|\mathcal{I}|}$. It was shown by \cite{BrownBook} that $\mathcal{P}_{\boldsymbol \gamma}$ is an exponential family with the canonical parameters ${\boldsymbol \gamma}$.

\begin{theorem}\label{ThCanORNEW}
The canonical parameters of $\mathcal{P}_{\boldsymbol \gamma}$ are the generalized log odds ratios in terms of ${\boldsymbol \delta}$.
\end{theorem}

\begin{proof}
Since the matrix $\mathbf{M}$ is full rank, then 
\begin{equation}\label{mixedPORnew}
\boldsymbol \gamma = (\mathbf{M}')^{-1}\mbox{log } \boldsymbol \delta.
\end{equation}
Let $\mathbf{B}$ denote the adjoint matrix to $\mathbf{M}'$ and write $\boldsymbol{b}_1, \dots, \boldsymbol{b}_{|\mathcal{I}|}$ for the rows of $\mathbf{B}$. The components of $\boldsymbol \gamma$ can be expressed as:
\begin{equation}\label{orbetas}
\gamma_i = \frac{1}{\mbox{det}(\mathbf{M})}\mbox{log }\boldsymbol \delta^{\boldsymbol{b}_i}, \,\, \mbox{for } i = 1, \dots, |\mathcal{I}|.
\end{equation}
All rows of $\mathbf{B}$ are integer vectors and thus the components of $\boldsymbol \gamma$ are multiples of the generalized log odds ratios. The common factor ${1}/{\mbox{det}(\mathbf{M})}\ne 0$ can be included in the canonical statistics, and the canonical parameters become equal to the generalized log odds ratios.
\end{proof}

Let $\mathbf{A}$ be a full row rank $J \times |\mathcal{I}|$ matrix with non-negative integer entries, and $\mathbf{D}$ denote a kernel basis matrix of $\mathbf{A}$. Set 
\begin{equation}\label{MatrixM}
\mathbf{M}=\left[\begin{array}{c} \mathbf{A}\\
\mathbf{D} \end{array}\right],
\end{equation}
find the inverse of $\mathbf{M}$ and partition it as 
$$\mathbf{M}^{-1}=\left[\mathbf{A}^-,\mathbf{D}^-\right].$$
Since $\mathbf{D}\mathbf{A}'=\boldsymbol 0$, then $(\mathbf{D}^-)'\mathbf{A}^- = \boldsymbol 0$. This matrix $\mathbf{M}$ can be used to derive a mixed parameterization of $\mathcal{P}$ with variation independent parameters \citep[cf.][]{BrownBook, HJ2}. Under this parameterization,
\begin{equation}\label{expandMix}
\boldsymbol \delta \longmapsto \left(\begin{array}{c} \boldsymbol \zeta_1  \\ \boldsymbol \zeta_2 \end{array}\right),
\end{equation}
where $\boldsymbol \zeta_1 = \mathbf{A}\boldsymbol \delta$ (mean-value parameters) and $\boldsymbol \zeta_2 = \mathbf{D}^-\mbox{log }\boldsymbol \delta$ (canonical parameters), and the range of the vector $(\boldsymbol \zeta_1, \boldsymbol \zeta_2)'$ is the Cartesian product of the separate ranges of $\boldsymbol \zeta_1$ and $\boldsymbol \zeta_2$.

Another mixed parameterization, which does not require calculating the inverse of $\mathbf{M}$, may be obtained as follows. 
Notice first that for any $\boldsymbol\delta \in \mathbb{R}^{|\mathcal{I}|}_+$ there exist unique vectors $\boldsymbol\beta \in \mathbb{R}^J$ and $\boldsymbol\theta \in \mathbb{R}^{|\mathcal{I}|-J}$ such that 
\begin{equation}\label{mixedPOR}
\mbox{log }\boldsymbol \delta = {\mathbf{A}'}\boldsymbol\beta + \mathbf{D}'\boldsymbol \theta.
\end{equation}
By orthogonality,
\begin{eqnarray}\label{thetaOR}
\mathbf{D}\mbox{log }\boldsymbol \delta &=& \boldsymbol 0 + \mathbf{D}\mathbf{D}'\boldsymbol \theta   \nonumber \\
\boldsymbol \theta &=& (\mathbf{D}\mathbf{D}')^{-1} \mathbf{D}\mbox{log }\boldsymbol \delta
\end{eqnarray}
Because of the uniqueness, $\mathbf{D}^- = (\mathbf{D}\mathbf{D}')^{-1} \mathbf{D}$. Moreover, since there is one-to-one correspondence between ${\boldsymbol \zeta}_2$ and $\tilde{\boldsymbol \zeta}_2 = \mathbf{D}\mbox{log }\boldsymbol \delta$, then, in the mixed parameterization, the parameter ${\boldsymbol \zeta}_2$ can be replaced with $\tilde{\boldsymbol \zeta}_2$. 
The components of $\tilde{\boldsymbol \zeta}_2 = \mathbf{D}\mbox{log }\boldsymbol \delta$ are some generalized log odds ratios  as well. 

A relational model is clearly defined and parameterized in the mixed parameterization derived from the model matrix of this model. In this parameterization the model requires logs of the generalized odds ratios to be zero and distributions in this model are parameterized by the remaining mean-value parameters. 

The following two examples illustrate the proposed mixed parameterization.

\textbf{ Example  \ref{CIex} (Revisited)}
{Consider a $2 \times 2\times 2$ contingency table and matrices $\mathbf{A}$ and $\mathbf{D}$ as in Example \ref{CIex}.
From (\ref{mixedPOR}): \begin{eqnarray} \label{mixed22}
\mbox{log }\boldsymbol p &=& {\mathbf{A}'}\boldsymbol\beta + \theta_1\cdot (1,0 ,-1,0, -1,0, 1 , 0)' + \theta_2\cdot (0,1,0, -1,0, -1,0, 1)',  
\end{eqnarray}
for some $\,\boldsymbol\beta \in \mathbb{R}^{6}$ and $\theta_1, \theta_2 \in \mathbb{R}$. 

Since the rows of $\mathbf{D}$ are mutually orthogonal, then   
\begin{eqnarray*}
(1,0 ,-1,0, -1,0, 1 , 0) \mbox{log }\boldsymbol p &=&  4\theta_1, \\
 (0,1,0, -1,0, -1,0, 1)  \mbox{log }\boldsymbol p &=&  4\theta_2. 
\end{eqnarray*} 
Thus, $\theta_1 = \frac{1}{4}\mbox{log } (p_{111}p_{221})/(p_{121}p_{211})$ and $\theta_2 = \frac{1}{4}\mbox{log } (p_{112}p_{222})/(p_{122}p_{212})$, as it is well known \citep[see e.g.][]{BFH}.

The parameters $\boldsymbol \beta$ can be expressed as generalized log odds ratios by applying (\ref{orbetas}): 
\begin{eqnarray*}
\beta_1 &=& \mbox{log } \frac{p_{111}^3p_{121}p_{211}}{p_{221}}, \hspace{20mm} 
\beta_2 = \mbox{log }\frac{p_{211}^2p_{221}^2}{p_{111}^2p_{121}^2},\\
\beta_3 &=& \mbox{log }\frac{p_{121}^2p_{221}^2}{p_{111}^2p_{211}^2}, \hspace{26mm}
\beta_4 = \mbox{log }\frac{p_{112}^3p_{122}p_{212}p_{221}}{p_{111}^3p_{121}p_{211}p_{222}},\\
\beta_5 &=& \mbox{log }\frac{p_{111}^2p_{121}^2p_{212}^2p_{222}^2}{p_{112}^2p_{122}^2p_{211}^2p_{221}^2}, \hspace{12mm}
\beta_6 = \mbox{log }\frac{p_{111}^2p_{122}^2p_{211}^2p_{222}^2}{p_{112}^2p_{121}^2p_{212}^2p_{221}^2}. 
\end{eqnarray*}
The mean-value parameters for this family are $\boldsymbol \zeta_1 = N\mathbf{A}\boldsymbol p$ (the expected values of the subset sums). The mixed parameterization consists of the mean-value parameters and the canonical parameters ${\boldsymbol \zeta}_2 = (\theta_1, \theta_2)'$ or $\tilde{\boldsymbol \zeta}_2 = \mathbf{D}\mbox{log }\boldsymbol p$.
\qed
}

Some models, more general than relational models, can be specified by setting generalized odds ratios equal to positive constants. An example of such model is given next. 

\begin{example} \label{HardyWeinbergEx}

The Hardy-Weinberg distribution arising in genetics was discussed as an exponential family in \cite{Barndorff1978, BrownBook}, among others.  Assume that a parent population contains alleles $G$ and $g$ with probabilities $\pi$ and $1-\pi$ respectively. The number of genotypes $GG$, $Gg$, and $gg$, that appear in a generation of $N$ descendants, is a random variable with $M(N, \boldsymbol p)$ distribution. Under the model of random mating and no selection, the vector of probabilities $\boldsymbol p$ has components
\begin{equation} \label{HWprobs}
p_1 = \pi^2, \,\, p_2 = 2\pi(1-\pi), \,\, p_3 = (1-\pi)^2.
\end{equation}
The model (\ref{HWprobs}) is a one-parameter regular exponential family with the canonical parameter $\mbox{log } \frac{\pi}{1-\pi}$. This model is slightly more general than relational models, but the techniques used for relational models apply. The model representation in terms of homogeneous odds ratios is 
\begin{equation}\label{HardyWmodel}
\frac{p_2^2}{p_1p_3} = 4.
\end{equation}
If the kernel basis matrix is chosen as $\mathbf{D} = (-1, 2 , -1)$ and the model matrix is  
$$\mathbf{A} =\left(\begin{array}{ccc}2 & 1 & 0\\ 0& 1 & 2 \end{array} \right),$$
the model (\ref{HardyWmodel}) can be expressed as
$$\mathbf{D} \mbox{log } \boldsymbol p = 2\mbox{log }2.$$

There exists a mixed parameterization of the family of multinomial distributions of the form 
\begin{equation}
\mbox{log } \boldsymbol p = \mathbf{A}' \boldsymbol \beta + \mathbf{D}' \theta.
\end{equation}
Here $\boldsymbol \beta = (\beta_1, \beta_2)'$ and $\theta \in (-\infty, \infty)$.
From the equation (\ref{thetaOR}): $$\theta = \frac{1}{6}\mbox{log }\frac{p_2^2}{p_1p_3}.$$ 
The parameter $\theta$ may be interpreted as a measure of the strength of selection in favor of the heterozygote character $Gg$ \cite[cf.][]{BrownBook}.

The condition $\mathbf{D} \mbox{log } \boldsymbol p = \mbox{log }4$ is equivalent to setting the parameter $\theta$ equal to  
$\frac{1}{6}\mbox{log } \frac{1}{4}.$
\qed
\end{example}

It is well known for a multidimensional contingency table that marginal distributions are variation independent from conditional odds ratios. Properly selected conditional odds ratios and sets of marginal distributions determine the distribution of the table uniquely  \citep{Barndorff1976, RudasSAGE, BergsmaRudas2003}. A generalization of this fact to the set $\mathcal{I}$ is given in the following theorem.   

\begin{theorem} \label{VIth}
Let $\mathcal{P}$ be the set of all positive distributions on the table $\mathcal{I}$. Suppose $\mathbf{A}$ is a non-negative integer matrix of full row rank and $\mathbf{D}$ is a kernel basis matrix of $\mathbf{A}$. Then the following statements hold: 
{\begin{enumerate}[(i)] 
\item For any ${P}_{\boldsymbol{\delta}_1},\,{P}_{\boldsymbol{\delta}_2} \in \mathcal{P}$ there exists a distribution ${P}_{\boldsymbol \delta} \in \mathcal{P}$ and a scalar $\alpha$ such that
$$\mathbf{A} \boldsymbol{\delta}  = \alpha \mathbf{A} \boldsymbol{\delta}_1 \,\, \mbox{and } \,\, \mathbf{D} \mbox{log }\boldsymbol \delta = \mathbf{D} \mbox{log } \boldsymbol{\delta}_2.$$
\item The coefficient of proportionality $\alpha = 1$ if and only if $\boldsymbol 1 \in R(\mathbf{A})$. 
\end{enumerate}}
\end{theorem}
The proof is straightforward, by Corollaries \ref{CorEq} and \ref{CorProp}, and is omitted here.

\section{Applications}\label{SecApplic}

The first example features relational models as a potential tool for modeling social mobility tables. A model of independence is considered on a space that is not the Cartesian product of the domains of the variables in the table. 

\begin{example}\label{DiagCellEx}
Social mobility tables often express a relation between statuses of two generations, for example, the relation between occupational statuses of respondents and their fathers, as in Table \ref{SocM} \citep{BlauDuncan}. To test the hypothesis of independence between respondent's mobility and father's status, consider the Respondent's mobility variable with three categories: Upward mobile (moving up compared to father's status), Immobile (staying at the same status), and Downward mobile (moving down compared to father's status). The initial table is thence transformed into Table \ref{FathMobTab}. 
\begin{table}[h]
\centering
\caption{Occupational Changes in a Generation, 1962}
\label{SocM}
\begin{tabular}{|c|c|c|c|}
\hline
Father's occupation &
\multicolumn{3}{|c|}{Respondent's occupation}\\
\cline{2-4}
&White-collar&Manual&Farm\\
\hline
White-collar	&6313	&2644&	132\\
\hline
Manual	&6321&	10883&	294\\
\hline
Farm&	2495&	6124&	2471	\\
\hline
\end{tabular}
\end{table}

\begin{table}[h]
\centering
\caption{Father's occupation vs Respondent's mobility. The MLEs are shown in parentheses.}
\label{FathMobTab}
\vspace{4mm}
\begin{tabular}{|c|c|c|c|}
\hline
Father's occupation &
\multicolumn{3}{|c|}{Respondent's mobility}\\
\cline{2-4}
&Upward&Immobile&Downward\\
\hline
White-collar	&-	&6313\, (7518.17)&2776\, (1570.83)\\
\hline
Manual	&6321\, (8823.66)&	10883\, (7175.18)&	294\, (1499.17)\\ 
\hline
Farm&	8619\, (6116.34)&	2471 \,(4973.66)&-\\ 
\hline
\end{tabular}
\end{table}

Since respondents cannot move up from the highest status or down from the lowest status, then the cells $(1,1)$ and $(3,3)$ in Table \ref{FathMobTab} do not exist. The set of cells $\mathcal{I}$ is a proper subset of the Cartesian product of the domains of the variables in the table. Let $\mathbf{S}$ be the class consisting of the cylindrical sets associated with the marginals, including the empty one. The relational model generated by $\mathbf{S}$ has the model matrix
$$
\mathbf{A}=\left[\begin{array}{ccccccc}
1 & 1 & 1 & 1 & 1 & 1 & 1\\
1 & 1 & 0 & 0 & 0 & 0 & 0\\
0 & 0 & 1 & 1 & 1 & 0 & 0\\
0 & 0 & 1 & 0 & 0 & 1 & 0\\
1 & 0 & 0 & 1 & 0 & 0 & 1\\
\end{array}\right]
$$
and is expressed in terms of local odds ratios as follows:
$$ \frac{p_{12}p_{23}}{p_{13}p_{22}} = 1, \,\,\, \frac{p_{21}p_{32}}{p_{22}p_{31}} = 1.$$
This model is a regular exponential family of order 4; the maximum likelihood estimates of cell frequencies exist and are  unique. (The estimates are shown in Table \ref{FathMobTab} next to the observed values.) The observed $X^2=6995.83$ on two degrees of freedom provides an evidence of strong association between father's occupation and respondent's mobility. \qed
\end{example}

The next example illustrates the usefulness of relational models for network analysis.  
\begin{example}\label{NetEx}

Table \ref{TradeEx} shows the total trade data between seven European countries that was collected from \cite{dataNet}. Every cell contains the value of trade volume for a pair of countries; cell counts are assumed to have Poisson distribution. The two  hypotheses of interest are: countries with larger economies generate more trade, and trade volume between two countries is higher if they use the same currency. In this example, GDP (gross domestic product) is is chosen as the characteristic of economy and Eurozone membership is chosen as the common currency indicator. The class  $\mathbf{S}$ includes five subsets of cells reflecting the GDP size:
$$\begin{array}{l}
\{GDP < 0.1 \cdot 10^6\, \mbox{ vs } \,GDP < 0.1 \cdot 10^6\},\\
 \{ GDP < 0.1 \cdot 10^6\, \mbox{ vs } \, 0.1 \cdot 10^6 \leq GDP < 0.6 \cdot 10^6\},\\ 
\{ GDP < 0.1 \cdot 10^6\, \mbox{ vs } \, GDP \geq 0.6 \cdot 10^6 \},\\ 
\{0.1 \cdot 10^6\leq GDP < 0.6 \cdot 10^6\, \mbox{ vs } \,0.1 \cdot 10^6 \leq GDP < 0.6 \cdot 10^6  \},\\ \{0.1 \cdot 10^6\leq GDP < 0.6 \cdot 10^6\, \mbox{ vs } \, GDP \geq 0.6 \cdot 10^6 \},
\end{array}$$
and three subsets reflecting Eurozone membership:
$$\begin{array}{l}
    \{\mbox{cells, showing trade between two Eurozone members }\},\\
     \{\mbox{cells, showing trade between a Eurozone member and a non-member }\},\\ 
    \{ \mbox{cells, showing trade between two Eurozone non-members}\}.
\end{array}$$
Under the model generated by $\mathbf{S}$, trade volume is the product of the GDP effect and the Eurozone membership effect.
 
\begin{table}[!h]
\centering
\caption{Total trade between seven countries (in billions US dollars). The MLEs are shown in parentheses.}
\label{TradeEx}
\vspace{5mm}
\begin{tabular}{|c|ccccccc|}
\hline
    &LV& NLD & FIN & EST & SWE &BEL & LUX\\
\hline
LV &[0]  &  0.7 (3.29)  & 1 (1.17) & 2 (2.0) & 1.3 (1.17)& 0.4 (1.17)  & 0.01 (0.01)\\   
NLD & - & [0] & 10 (17) & 1 (1.17) & 17 (15) & 102 (102) & 2.1 (2.29)\\  
FIN& - & -  & [0] & 4 (1.17) & 18 (15) & 4 (2.29) & 0.1 (2.29)\\  
EST & - & - & - & [0] &  2.6 (1.17) & 0.5 (1.17)  & 0.01 (0.01)\\  
SWE & - & - & - & -  & [0] &  15 (15)  & 0.35 (2.29)\\ 
BEL & - & - & -&  - &- &[0] &   9 (6.41) \\ 
LUX & - & - & -& -&- &- & [0] \\
\hline
\end{tabular}                  
\end{table}

This model is a regular exponential family of order 6. The maximum likelihood estimates for cell frequencies exist and are unique. The observed $X^2 = 20.16$ on 14 degrees of freedom yields the asymptotic p-value of $0.12$; so the model fits the trade data well. Alternatively, sensitivity of the model fit to other choices regarding GDP could also be studied. \qed
\end{example}

\section*{Acknowledgments}

The authors wish to thank Thomas Richardson for his comments and discussions.

The work of Anna Klimova and Adrian Dobra was supported in part by NIH Grant R01 HL092071. 

Tam\'{a}s Rudas was supported in part by Grant No TAMOP 4.2.1./B-09/1/KMR-2010-0003 from the European Union and the European Social Fund. He is also a Recurrent Visiting Professor at the Central European University and the moral support received is acknowledged.

\bibliographystyle{apacite}
\bibliography{RDK}

\end{document}